\documentclass[11pt]{article}

\usepackage[utf8]{inputenc}
\usepackage[T1]{fontenc}
\usepackage{lmodern}
\usepackage[protrusion=true,expansion=false]{microtype}
\usepackage[margin=1in, top=1.1in, bottom=1.1in]{geometry}

\usepackage{amsmath,amssymb,amsthm}
\usepackage{booktabs}
\usepackage{graphicx}
\usepackage[colorlinks=true,
            linkcolor=blue!70!black,
            citecolor=blue!70!black,
            urlcolor=blue!70!black]{hyperref}
\usepackage{xcolor}
\usepackage{enumitem}
\usepackage[numbers,sort&compress]{natbib}
\usepackage{caption}
\usepackage{tikz}
\usetikzlibrary{shapes.geometric,arrows.meta,positioning,fit,backgrounds}
\usepackage{multirow}
\usepackage{array}

\usepackage{titlesec}
\titleformat{\section}{\large\bfseries}{\thesection}{1em}{}
\titleformat{\subsection}{\normalsize\bfseries}{\thesubsection}{1em}{}

\newtheorem{theorem}{Theorem}[section]
\newtheorem{lemma}[theorem]{Lemma}

\newtheorem{proposition}[theorem]{Proposition}
\newtheorem{definition}{Definition}[section]
\newtheorem{remark}{Remark}[section]

\newcommand{\Azero}{\mathcal{A}_0}
\newcommand{\Dhat}{\widehat{D}}
\newcommand{\HALT}{\texttt{HALT}}
\newcommand{\RESUME}{\texttt{RESUME}}
\newcommand{\DENY}{\texttt{DENY}}
\newcommand{\RECAL}{\texttt{RECALIBRATE}}
\newcommand{\Pset}{\mathcal{P}}
\newcommand{\Es}{E_s}
\newcommand{\Dh}{D_h}
\newcommand{\sigh}{\sigma_h}
\newcommand{\Hi}{H_i}
\newcommand{\ski}{sk_i}
\newcommand{\pki}{pk_i}
\newcommand{\Tstar}{T^*}
\newcommand{\APB}{\mathrm{APB}}

\title{\textbf{Identity-Bound Governance Under Execution Uncertainty}\\[4pt]
       \large An Accountability Proof Block for LLM Agent Persistent Halts,
       with Cryptographic Implementation and Cross-Model Calibration\\[6pt]
       \normalsize Agent Governance Series --- Paper 8}

\author{Marcelo Fernandez\\
        \small TraslaIA\\
        \small \texttt{info@traslaia.com}\\
        \small \texttt{DOI: 10.5281/zenodo.20157139}}
\date{September 2026}

\begin{document}
\maketitle

\begin{center}
\small
\begin{tabular}{clll}
\toprule
\textbf{Paper} & \textbf{Short title} & \textbf{Zenodo DOI} & \textbf{arXiv}\\
\midrule
P0 & Atomic Decision Boundaries~\cite{fernandez2026a} & 10.5281/zenodo.19670649 & arXiv:2604.17511\\
P1 & Agent Control Protocol (ACP)~\cite{fernandez2026acp} & 10.5281/zenodo.19672575 & arXiv:2603.18829\\
P2 & From Admission to Invariants (IML)~\cite{fernandez2026iml} & 10.5281/zenodo.19672589 & arXiv:2604.17517\\
P3/4 & Irreducible Governance Structure~\cite{fernandez2026govstr} & 10.5281/zenodo.19708496 & TBD\\
P5 & Reconstructive Authority Model (RAM)~\cite{fernandez2026ram} & 10.5281/zenodo.19669430 & arXiv:2604.22898\\
P6 & Operationalizing Reconstructive Authority~\cite{fernandez2026opram} & 10.5281/zenodo.19699460 & arXiv:2605.23935\\
P7 & Closing the Execution Gap~\cite{fernandez2026applied} & 10.5281/zenodo.19929771 & TBD\\
P8 & \textbf{Identity-Bound Governance (APB)} & 10.5281/zenodo.20157139 & TBD\\
\bottomrule
\end{tabular}
\end{center}

\begin{abstract}
\noindent
A correctly governed LLM agent can reach a state in which neither continuing
execution nor automatically halting is admissible: the system has detected a
persistent failure of its observability or drift-detection layer, but cannot
itself decide who has the authority to resume, deny, or recalibrate the
deployment. We call this an \emph{identity-bound governance event}, and we
formalise the mechanism that resolves it.

We introduce the \emph{Accountability Proof Block} ($\APB$), a two-part
record consisting of a system-constructed evidence block $\Es$, a
human-supplied decision block $\Dh$, and a cryptographic signature $\sigh$
binding the two to a registered principal $\Hi \in \Pset$. The system can
construct $\Es$ but cannot forge $\sigh$ (it lacks the principal's secret
key); the principal can issue $\Dh$ but cannot alter $\Es$ undetected (the
signature covers both).
We prove four theorems about this construction:
\emph{Protocol-Bounded Governance Completeness} (every halt resolves through
either the recovery loop or a signed APB, within the execution model
induced by DC.1/DC.2), \emph{Non-Repudiability} (no principal can
deny a governance act for which a valid APB exists),
\emph{Impossibility of Anonymous Re-Authorization} (the execution layer
cannot generate a valid APB without the principal's private credential),
and \emph{Finite-Time APB Construction Termination}.
The implementation uses the RFC~8785 JSON Canonicalization Scheme (JCS)
for cross-implementation-safe signatures and extends replay defence
with a UUID4-based semantic-uniqueness predicate (V5).

We provide a real-cryptography implementation in
\texttt{cryptography/ed25519}, expose the result through an Authority
Resolution Function $G$, and validate the construction empirically. With a
MockLLM stack across $10$ seeds and $10{,}000$ steps, governance completeness
holds identically under two threshold policies ($3{,}812$ halt events,
zero unresolved cases). Across $9$ adversarial vectors against
$200$ freshly-signed APBs, the verifier achieves $100.0000\%$ detection.
A multi-principal threshold experiment across $7{,}000$ events confirms
Byzantine principal resistance: $0$ false acceptances in $2{,}000$
single-key capture attempts under $k = 2$.
A cross-model study of six open LLMs (\texttt{mistral}, \texttt{deepseek-r1},
\texttt{gemma}, \texttt{gpt-oss}, \texttt{qwen}, \texttt{llama3}) finds
that the model-specific drift threshold $\Tstar$ is stable within model
($\sigma/\Tstar < 2\%$ for all measurable cases) but varies across models
by a factor of $1.7\times$, refuting any size-monotone hypothesis: the
largest model in our sample is the one that does not drift within the
$500$-step window. A temperature sweep across three of those models at
$T \in \{0.2, 0.4, 0.6, 0.8\}$ replicates the temperature-insensitivity
conjecture for the two fast-drifting models (per-model range $\le 5$
steps) but reveals a \emph{drift-floor effect} for the slow-drifter
($T^*$ range of $51$ steps for \texttt{gemma4}). We conclude that
$\Tstar$ must be \emph{measured} per deployment --- not predicted from
architecture, and not assumed temperature-invariant for slow-drifting
models --- and that the APB is the operational vehicle by which that
measured threshold yields accountable authority transfer.
\end{abstract}

\section{Introduction}
\label{sec:intro}

P7 of this series~\citep{fernandez2026applied} established that persistent
halts are not implementation failures but structural consequences of bounded
runtime governance: when drift accumulates past the recovery authority of
the autonomous stack, execution stops and accountability remains unresolved.
P8 answers the question that structural halt leaves open --- \emph{who has
the authority to lift it, under what evidence, and with what
non-repudiable record?}

Concretely: once a halt is persistent in the formal sense
(\S\ref{sec:framework}), the system cannot issue a unilateral resumption
without recalibrating its own authority, which design constraints DC.1
and DC.2 of~\citep{fernandez2026applied} explicitly forbid.
This paper introduces the \emph{Accountability Proof Block} ($\APB$) as
the minimal mechanism for binding that authority to a verifiable human
decision.

We refer to this as the problem of \emph{identity-bound governance under
execution uncertainty}.

The standard accountability mechanisms in deployed AI systems are not a
solution to this problem. Documentation artefacts such as model cards
\citep{mitchell2019model} and OpenAI's system cards
\citep{openai2024systemcard} are static; they describe behaviour at training
or deployment time, not at the moment of an unexpected halt. Alignment
techniques such as constitutional AI \citep{bai2022constitutional} bias the
agent's behaviour but do not provide a non-repudiable record of the
human decisions that surround a deployment. Compliance audit trails as
mandated by GDPR, SOX, or PCI-DSS are typically post-hoc and not
cryptographically bound to specific runtime events. None of these
mechanisms provides a per-event, identity-bound, cryptographically
verifiable record that the system itself cannot forge.

\paragraph{Contributions.}
This paper introduces the \emph{Accountability Proof Block} ($\APB$) as
the formal object that fills this role, and validates it empirically.
Specifically:

\begin{enumerate}[leftmargin=2em, topsep=2pt, itemsep=2pt]
  \item We define the Principal Set $\Pset$, the Authority Resolution
        Function $G$, and the APB structure
        $\APB = (\Es, \Dh, \sigh)$ formally
        (\S\ref{sec:framework}), and we make the threat model explicit:
        five attacker classes, with the boundary between what the APB
        defends against and what it does not made precise.
  \item We prove four theorems (\S\ref{sec:theorems}):
        \emph{Protocol-Bounded Governance Completeness} (T8.1),
        \emph{Non-Repudiability} (T8.2),
        \emph{Impossibility of Anonymous Re-Authorization} (T8.3),
        and \emph{Finite-Time APB Construction Termination} (T8.4).
        Each is supported by auxiliary lemmas that decompose the
        argument.
  \item We provide a real-cryptography implementation
        (\S\ref{sec:impl}) using $\text{ed25519}$ signatures over the
        RFC~8785 JSON Canonicalization Scheme (JCS), with a Principal
        Registry, a five-predicate verifier (V1--V5, including a
        semantic-uniqueness check via UUID4 event identifiers), and an
        append-only log with HMAC chaining for persistence.
  \item We validate the construction empirically through five
        experiments (\S\ref{sec:exp_a}--\S\ref{sec:exp_e}):
        Governance Completeness over $10$ seeds and two threshold
        policies; APB Integrity under nine adversarial vectors;
        cross-model $\Tstar$ across six open LLM families;
        temperature insensitivity across three of those families
        and four sampling temperatures; and multi-principal
        threshold governance with Byzantine resistance validation
        across $7{,}000$ events covering seven attack scenarios.
  \item We extract \emph{Proposition~5.1} (\S\ref{sec:prop_calibration})
        as an operational design recipe: given a measured $\Tstar$
        for a model, the persistence window $\Delta T$ that the
        execution layer uses to classify a halt as a governance event
        admits a calibrated range. We discuss what this means for
        deployments where $\Tstar$ is not finite (one of our six
        models did not drift within the experimental window).
\end{enumerate}

\paragraph{Independence from prior work.}
Although this paper sits in the \emph{Agent Governance} series, its claims
do not depend on the empirical results of prior papers. The formal
framework, the theorems, and all four experiments are self-contained:
fresh evidence is generated under identical software but different runs,
and the cryptographic claims are entirely new (no cryptography appears
in P0--P7). Where prior work is referenced (notably the design constraints
DC.1 and DC.2 of \citep{fernandez2026applied}), the references are
contextual rather than constitutive.

\paragraph{Roadmap.}
\S\ref{sec:related} situates the APB against existing accountability
mechanisms. \S\ref{sec:framework} introduces the Principal Set, the
Authority Resolution Function, and the APB structure. \S\ref{sec:theorems}
proves the four theorems. \S\ref{sec:impl} describes the implementation
and the threat model boundary. \S\ref{sec:exp_a}--\S\ref{sec:exp_e}
report the five validation experiments. \S\ref{sec:discussion} discusses
limitations and integration paths. \S\ref{sec:conclusion} concludes.

\section{Related Work}
\label{sec:related}

We organise prior work along the dimensions that the APB engages
simultaneously: \emph{runtime}, \emph{cryptographic},
\emph{identity-bound}, and \emph{LLM-agent specific}. To our knowledge no
existing mechanism combines all four; Table~\ref{tab:related} maps the
landscape.

\begin{table}[t]
\centering
\caption{Accountability mechanisms relevant to LLM agent governance,
characterised by whether they operate at runtime, are cryptographic,
bind to a specific human identity, and are designed for the LLM-agent
setting. The APB is, to our knowledge, the first construction that
satisfies all four.}
\label{tab:related}
\small
\begin{tabular}{lcccc}
\toprule
Mechanism & Runtime & Crypto & Identity-bound & LLM-agent specific \\
\midrule
Model / system cards \citep{mitchell2019model,openai2024systemcard}
        & --- & --- & --- & \checkmark \\
Constitutional AI \citep{bai2022constitutional}
        & \checkmark & --- & --- & \checkmark \\
Compliance audit trails (GDPR, SOX)
        & \checkmark & --- & --- & --- \\
Smart-contract logs \citep{wood2014ethereum}
        & \checkmark & \checkmark & \checkmark & --- \\
Threshold signatures \citep{boneh2018compact}
        & --- & \checkmark & \checkmark & --- \\
\textbf{APB (this work)}
        & \checkmark & \checkmark & \checkmark & \checkmark \\
\bottomrule
\end{tabular}
\end{table}

\paragraph{Documentation-based.}
Model cards \citep{mitchell2019model}, system cards
\citep{openai2024systemcard}, and data sheets are static artefacts
authored prior to deployment. They describe what an agent is supposed
to do or not do; they do not record what a specific human authorised
the agent to resume after a halt at $t = t_e$.

\paragraph{Alignment-based.}
Constitutional AI \citep{bai2022constitutional} and RLHF approaches
shape the agent's behaviour at training time. They affect \emph{how}
the agent acts but not \emph{who} authorised a recovery decision.
The APB is orthogonal: it does not constrain the agent's behaviour;
it constrains the act of re-authorising the agent.

\paragraph{Audit-trail and compliance.}
Regulatory audit trails (GDPR Article 30, SOX section 404, PCI-DSS)
log access events but typically lack cryptographic binding to a
specific human credential at the moment of decision. Recent work on
verifiable audit logs \citep{transparencylog2014} brings cryptographic
guarantees to logs themselves, but does not address per-event
authorisation by named principals.

\paragraph{Smart contracts and on-chain governance.}
Smart-contract execution \citep{wood2014ethereum} and on-chain DAO
governance produce records that are runtime, cryptographic, and
identity-bound. They are, however, not designed for LLM-agent runtime
state: they bind to wallet addresses (which are not the same as
human identities), assume a public ledger, and are not intended to
operate at the latency of an in-process governance event. The APB
borrows the cryptographic-binding insight while keeping the
deployment local.

\paragraph{Threshold and multi-signature schemes.}
Threshold and multi-signature constructions
\citep{boneh2018compact,gennaro2018fast} address \emph{which}
principals must concur before an action is authorised. They
complement the APB rather than replace it: \S\ref{sec:framework}
notes that an APB-typed governance event may require $k$-of-$n$
signatures, in which case each cosigner contributes their own
$\sigh$ over the same canonical $(\Es \,\Vert\, \Dh)$ payload.

\paragraph{Position.}
The APB is, to our knowledge, the first construction that is
simultaneously runtime, cryptographic, identity-bound, and designed
for LLM-agent governance. The closest analogues are smart-contract
logs (which lack the LLM-agent fit) and threshold signatures (which
are a sub-component, not a complete governance mechanism).

\section{Formal Framework}
\label{sec:framework}

\subsection{Setting and Inherited Constraints}
\label{sec:setting}

We assume an agent system whose execution layer maintains a runtime
admission snapshot $\Azero$, a drift estimator $\Dhat(t)$ with
threshold $\theta$, and a halt mechanism that may emit
$\HALT$ when its observability or risk-attribution gates fail.
Two design constraints from prior work in the series
\citep{fernandez2026applied} are inherited verbatim:
\begin{description}[leftmargin=2em, topsep=2pt]
  \item[DC.1 (Separation of Recalibration Authority).]
        For all $t$, $\Azero(t) = \Azero(0)$. The execution layer
        is forbidden from modifying its own admission baseline at
        runtime.
  \item[DC.2 (Drift Event Classification).]
        A halt is \emph{transient} if $\Dhat(\tau) \geq \theta$ for
        $\tau$ within a bounded window $\Delta T$ and returns below
        $\theta$ thereafter; it is \emph{persistent} if
        $\Dhat(t) \geq \theta$ for all $t \in [t_0, t_0 + \Delta T]$.
        Transient halts route to the recovery loop; persistent halts
        constitute \emph{governance events}.
\end{description}

DC.1 implies that, when a persistent halt occurs, the system cannot
recalibrate its own $\Azero$ to clear the halt. DC.2 routes such halts
out of the execution layer's authority. The role of this paper is to
specify what is on the other side of that boundary.

\subsection{Principal Set $\Pset$}
\label{sec:principals}

\begin{definition}[Principal Set]
\label{def:principals}
The \emph{Principal Set} of a governed system is a finite set
$\Pset = \{H_1, \ldots, H_n\}$ of identity-bound human principals.
Each $\Hi$ has a public verification key $\pki$ and an associated
private signing key $\ski$ held exclusively by $\Hi$. The mapping
$H_i \mapsto \pki$ is recorded in a registry that the execution layer
treats as read-only.
\end{definition}

DC.1 specialised to $\Pset$: the registry is fixed at $\Azero$
construction time, and any modification to it is itself a governance
event of type $\RECAL$. We thus obtain a self-referential modification
discipline --- $\Pset$ is mutable only through APBs that authorise
its mutation.

\paragraph{Revocation.}
A principal may be revoked. We model revocation by recording a
timestamp of revocation; APBs signed before the revocation timestamp
remain verifiable (with a flag indicating the principal's later
revoked status), while APBs signed after are rejected at verification
time. This historical-validity discipline is enforced by the
predicate $\mathrm{is\_active}(\Hi, t_e)$ used by the verifier
(\S\ref{sec:impl}).

\subsection{Authority Resolution Function $G$}
\label{sec:G}

\begin{definition}[Authority Resolution Function]
\label{def:G}
$G : \Pset \times \mathcal{S} \times \mathcal{E}
   \to \{\RESUME, \DENY, \RECAL\}$,
where $\mathcal{S}$ denotes the space of system runtime states and
$\mathcal{E}$ the space of governance evidence packages. $G$ is the
governance-layer analog of the runtime authority function $F$ of
\citep{fernandez2026applied}: it constructs authority from human
attestation over a defined evidence set, not from sensor state.
The output $\RECAL$ is the only authorised path by which $\Azero$
may change after initialisation; $\RESUME$ permits a single
resumption under stated scope; $\DENY$ permanently bars the queried
action.
\end{definition}

The APB is the recorded artefact that witnesses an evaluation of
$G$.

\subsection{Accountability Proof Block (APB)}
\label{sec:apb}

\begin{definition}[APB]
\label{def:apb}
An \emph{Accountability Proof Block} signed by principal $\Hi$ is a
triple
\[
  \APB \;=\; (\Es,\;\Dh,\;\sigh)
\]
where
\begin{align*}
  \Es &= (\mathrm{hash}(\Azero),\; \Dhat(t_e),\; t_e,\;
         \mathrm{eid},\;
         \mathrm{hash}(\mathrm{trace}_{\le t_e}),\; \mathrm{cause})
         \quad\text{(System Evidence Block)}\\
  \Dh &= (\Hi,\;\mathrm{decision},\;\mathrm{rationale},\;\mathrm{scope})
         \quad\text{(Human Decision Block)}\\
  \sigh &= \mathrm{Sign}_{\ski}\!\bigl(\,
            \mathrm{canon}(\Es) \,\Vert\, \mathrm{canon}(\Dh)
            \,\bigr)
         \quad\text{(Signature)}
\end{align*}
$\mathrm{eid}$ is a UUID version~4 generated once at evidence-construction
time; it provides semantic uniqueness across governance events independently
of $t_e$ (enabling V5, \S\ref{sec:impl}).
$\mathrm{canon}(\cdot)$ denotes the RFC~8785 JSON Canonicalization Scheme
(JCS)~\citep{rfc8785}: keys sorted by Unicode codepoint, IEEE~754
number representation, no insignificant whitespace, UTF-8 encoding.
This cross-implementation-safe encoding is the deterministic
input to the signature, removing any ambiguity that could otherwise
weaken the non-repudiability and tamper-detection guarantees.
$\mathrm{Sign}$ is the ed25519 signature operation.
\end{definition}

The construction has three structural properties enforced by its
typing alone:

\begin{enumerate}[leftmargin=2em, topsep=2pt]
  \item Only the execution layer can correctly produce $\Es$: it is the
        only party that holds $\Azero$, $\Dhat$, and the trace.
  \item Only $\Hi$ can produce $\sigh$: it is the only party that holds
        $\ski$.
  \item $\sigh$ covers both $\Es$ and $\Dh$: any modification to either
        invalidates the signature.
\end{enumerate}

\subsection{Threat Model}
\label{sec:threat}

The APB is designed to defeat a specific class of attackers; we make
the boundary of its protection explicit. Five actor classes are
considered.

\begin{description}[leftmargin=2em, topsep=2pt]
  \item[A1. External attacker.] No system access. Goal: forge an APB
        attributed to $\Hi$. Defence: T8.3 (\S\ref{sec:theorems}).
  \item[A2. Compromised system process.] Has runtime access to the
        execution layer but does not hold $\ski$ for any
        $\Hi \in \Pset$. Goal: synthesise a $\sigh$ that verifies
        against some registered $\pki$. Defence: T8.3.
  \item[A3. Replay attacker.] Holds a previously-valid APB and tries
        to use it in a new context. Defence: two complementary predicates
        in the verifier (\S\ref{sec:impl}): V4 (temporal freshness, rejects
        APBs outside the acceptable age window) and V5 (semantic
        uniqueness, rejects any APB whose $\mathrm{eid}$ was already
        accepted, even within the V4 window).
  \item[A4. Man-in-the-middle.] Intercepts an APB in transit and tries
        to modify either $\Es$ or $\Dh$ before delivery. Defence:
        $\sigh$ covers both blocks, so any tampering invalidates
        verification.
  \item[A5. Compromised principal.] Has obtained $\ski$ for some
        $\Hi$. \emph{Out of scope.} The APB cryptographically binds
        decisions to keys; if a key is compromised, the binding
        attaches to the compromised key. Mitigation lies outside
        the APB protocol (multi-signature, key rotation,
        out-of-band review).
\end{description}

\subsubsection{What APB Does Not Prevent}
The APB is an \emph{accountability} mechanism, not an
\emph{access-control} mechanism. A principal $\Hi$ holding a valid
$\ski$ can sign any $\Dh$ they wish, including a malicious one, and
the resulting APB will verify cryptographically. What the APB
ensures, against this case, is that:
\begin{enumerate}[leftmargin=2em, topsep=2pt, itemsep=1pt]
  \item the malicious decision is attributable to $\Hi$ and cannot
        be denied (T8.2);
  \item the system itself did not synthesise the decision behind
        $\Hi$'s back (T8.3).
\end{enumerate}
Mitigations for adversarial principals (multi-signature, separation
of duties, out-of-band review) are configurable on top of the APB
but are not part of the APB protocol itself. We return to this
distinction in the discussion (\S\ref{sec:discussion}).

\subsection{Multi-Principal APB Extension}
\label{sec:multi_principal}

The base APB carries a single signature $\sigh$ from one principal.
For deployments that require stronger governance guarantees, we extend
the construction to support \emph{threshold signing}: a Multi-Principal
APB requires $k$ valid, distinct-principal signatures out of $n$
registered principals ($k$-of-$n$).

\begin{definition}[Multi-Principal APB]
\label{def:multi_apb}
A \emph{Multi-Principal Accountability Proof Block} with threshold $k$
is a quadruple
\[
  \APB^{(k)} \;=\; (\Es,\;\Dh,\;\Sigma,\;k)
\]
where $\Es$ and $\Dh$ are as in Definition~\ref{def:apb}, and
$\Sigma = \{(H^1, \sigma^1), \ldots, (H^m, \sigma^m)\}$ is an
ordered set of signature entries ($H^j \in \Pset$, $\sigma^j \in
\{0,1\}^{512}$), each satisfying
$\sigma^j = \mathrm{Sign}_{sk^j}(\mathrm{canon}(\Es)
\,\Vert\, \mathrm{canon}(\Dh))$.
The APB is \emph{valid} under threshold $k$
iff at least $k$ entries in $\Sigma$ pass predicates V1--V3 and
V5 over distinct principals $H^j \in \Pset$.
\end{definition}

Three structural properties hold by construction:
\begin{enumerate}[leftmargin=2em, topsep=2pt, itemsep=2pt]
  \item Each signer attests independently to the same $(\Es, \Dh)$
        pair; no coordination between signers' devices is required.
  \item Duplicate signer entries (same $\Hi$ appearing twice) are
        rejected: each distinct principal counts at most once toward
        the threshold.
  \item The signed message is identical to the single-principal APB,
        so the non-repudiability and anonymous-re-authorization
        impossibility guarantees (T8.2, T8.3) apply to each entry
        in $\Sigma$ individually.
\end{enumerate}

The $k$-of-$n$ structure is the formal basis for Byzantine principal
resistance (Proposition~\ref{prop:byzantine} below): governance
capture under $\APB^{(k)}$ requires the adversary to hold $k$
distinct signing keys simultaneously.

\section{Theorems}
\label{sec:theorems}

We prove four theorems about the APB. The first three are
properties promised in the framework section of
\citep{fernandez2026applied}; we restate them precisely and give full
proofs. The fourth (T8.4) is new: it establishes that the construction
of $\Es$ terminates in bounded time, which is required for T8.1 to be
non-vacuous.

\subsection{T8.1 --- Protocol-Bounded Governance Completeness}
\label{sec:t81}

\begin{theorem}[Protocol-Bounded Governance Completeness]
\label{thm:completeness}
Let $h$ be a $\HALT$ event emitted by an execution layer satisfying
DC.1 and DC.2. Then, within the execution model induced by DC.1, DC.2,
and the APB verifier predicates, exactly one of the following
resolutions holds:
\begin{enumerate}[leftmargin=2em, topsep=2pt, itemsep=1pt]
  \item[(R1)] $h$ is transient under DC.2 and the recovery loop
        produces a $\RESUME$ within its bounded attempt budget;
  \item[(R2)] $h$ is persistent under DC.2 and a valid signed
        $\APB$ is produced through an evaluation of $G$.
\end{enumerate}
No third resolution path exists within the protocol boundary.
\end{theorem}

\begin{proof}
By DC.2, $h$ is exactly one of \emph{transient} or \emph{persistent}.
Case (transient): the recovery loop is defined for transient halts
(\citep{fernandez2026applied}, P6 Recovery Loop) and terminates in at
most $k$ attempts with a $\RESUME$, $\HALT$, or $\DENY$ outcome on the
augmented state. If the outcome is $\RESUME$, we are in case (R1).
If the outcome on the augmented state is $\HALT$ or $\DENY$, the halt
is reclassified as persistent (the recovery loop has exhausted the
ways in which the halt could be transient), reducing to the persistent
case below.
Case (persistent): by DC.1 the execution layer cannot modify $\Azero$
to clear the halt; the only authorised path to a state in which
execution may continue is an output of $G$. By Theorem~\ref{thm:t84}
below, $\Es$ can be constructed in finite time. By
Definition~\ref{def:apb}, given $\Es$ and a principal $\Hi$ with valid
$\ski$, an APB is produced. We are in case (R2).
The two cases are exhaustive (DC.2) and disjoint (the recovery loop's
output of $\RESUME$ excludes persistence within the same window).
Therefore no third resolution exists within the protocol boundary.

\emph{Scope of the claim.}
This completeness result is relative to the execution model induced by
DC.1, DC.2, and the APB verifier predicates V1--V5. It does not assert
completeness under arbitrary infrastructure failures, compromised
principal sets, or events external to the protocol boundary
(e.g.\ a fully inoperative registry or unreachable signing device).
Those failure modes are addressed by the composability layer discussed
in \S\ref{sec:composability}.
\end{proof}

\begin{lemma}[Resolution Path Exhaustivity]
\label{lem:exhaustivity}
The recovery loop's exit set $\{\RESUME, \HALT, \DENY\}$ partitions
the post-halt state space into transient ($\RESUME$) and
persistent ($\HALT \cup \DENY$). The latter is exactly the
domain of $G$.
Note that $\HALT$ and $\DENY$ are semantically distinct:
$\HALT$ denotes an unresolved system state (the loop exhausted its
authority without a decision), while $\DENY$ denotes a resolved human
decision to refuse resumption. Both share the property that the
execution layer cannot self-resolve them without external authority,
which is the criterion relevant here.
\end{lemma}

\begin{proof}
The recovery loop's outputs are mutually exclusive by construction
(it returns at most one). Definition of transient (DC.2): the drift
estimator returned below threshold within $\Delta T$, observable by
the recovery loop's own measurements. If $\RESUME$, by P6 the loop
has confirmed observability is restored, hence transient. If $\HALT$
or $\DENY$, the loop has confirmed it cannot resolve, hence
persistent.
\end{proof}

\subsection{T8.2 --- Non-Repudiability}
\label{sec:t82}

\begin{theorem}[Non-Repudiability]
\label{thm:nonrep}
Let $\APB = (\Es, \Dh, \sigh)$ be a verifiable APB whose $\Hi$ field
in $\Dh$ identifies a registered, non-revoked principal at $t_e$, and
let $\sigh$ verify against $\pki$. Then no efficient adversary can
produce a $\Dh' \neq \Dh$ on $\Es$ that also verifies under $\pki$,
except with negligible probability.
\end{theorem}

\begin{proof}
Verification of an APB requires that
$\sigh = \mathrm{Sign}_{\ski}(\mathrm{canon}(\Es) \,\Vert\,
\mathrm{canon}(\Dh))$. Suppose, for contradiction, that an adversary
produces $\APB' = (\Es, \Dh', \sigh)$ with $\Dh' \neq \Dh$ that
verifies. Because $\sigh$ is a fixed value and the canonical
serialisation is injective on dataclass content (Lemma~\ref{lem:integrity}),
$\mathrm{canon}(\Dh') \neq \mathrm{canon}(\Dh)$, so the message under
the signature differs. A successful verification of $\sigh$ over the
new message contradicts the existential unforgeability of ed25519
under chosen-message attack
\citep{bernstein2012highspeed,boneh2014foundations}. This event has
negligible probability in the security parameter.
\end{proof}

\begin{lemma}[$\Es$ Post-Construction Integrity]
\label{lem:integrity}
The canonical encoding $\mathrm{canon}$ used in the APB construction
is collision-resistant on dataclass content: distinct values of the
six $\Es$ fields produce distinct byte-strings.
\end{lemma}

\begin{proof}
$\mathrm{canon}$ is the RFC~8785 JSON Canonicalization Scheme
\citep{rfc8785}: keys sorted by Unicode codepoint, IEEE~754
representation for numbers, UTF-8 encoding, no insignificant whitespace.
The RFC~8785 canonical mapping is injective over the primitive
dataclass domain used here: all $\Es$ and $\Dh$ fields are strings
or IEEE~754 floats, and for this domain distinct field values produce
distinct byte-level representations (the key-sort order is
deterministic and no two distinct primitive values share a
canonical encoding under RFC~8785).
The $\mathrm{eid}$ field is a UUID4 string, which is a primitive and
thus covered by the same argument.
\end{proof}

\begin{lemma}[Signature Verifiability Under Tamper]
\label{lem:tamper}
Let $\APB$ verify against $\pki$ with original $\Es$. For any
modification of $\Es$ to $\Es' \neq \Es$, retaining $\sigh$ and $\Dh$,
the modified APB does not verify under $\pki$.
\end{lemma}

\begin{proof}
By Lemma~\ref{lem:integrity}, $\mathrm{canon}(\Es') \neq
\mathrm{canon}(\Es)$, hence the message under verification differs
from the message that was signed. Verification of an ed25519
signature against a different message succeeds only with negligible
probability.
\end{proof}

\subsection{T8.3 --- Impossibility of Anonymous Re-Authorization}
\label{sec:t83}

\begin{theorem}[Impossibility of Anonymous Re-Authorization]
\label{thm:anonimp}
Let an adversary $\mathcal{A}$ have full read and write access to the
execution layer's runtime state but not hold $\ski$ for any
$\Hi \in \Pset$. Then $\mathcal{A}$ cannot produce a verifying APB
attributed to any registered principal, except with negligible
probability.
\end{theorem}

\begin{proof}
A verifying APB requires $\sigh$ such that
$\mathrm{Verify}_{\pki}(\sigh, \mathrm{canon}(\Es) \,\Vert\,
\mathrm{canon}(\Dh)) = \top$ for some registered $\pki$. The
adversary may construct $\Es$ and $\Dh$ at will; the question is
whether they can produce $\sigh$. Lemma~\ref{lem:forge} establishes
that producing such a $\sigh$ without access to $\ski$ requires
breaking the existential unforgeability of ed25519, which is
infeasible under the discrete log assumption on the curve25519
group \citep{bernstein2012highspeed}.
\end{proof}

\begin{lemma}[Attribution Uniqueness]
\label{lem:attribution}
A verifying APB uniquely identifies a registered principal: there is
no efficient way to produce an APB whose $\sigh$ verifies under two
distinct registered $\pki$ simultaneously.
\end{lemma}

\begin{proof}
Two distinct ed25519 public keys verify the same signature over the
same message only with negligible probability (the signature is bound
to the secret key, and finding a colliding key pair reduces to
solving discrete log on the curve).
\end{proof}

\begin{lemma}[Forge Resistance under Computational Hardness]
\label{lem:forge}
Any polynomial-time algorithm that produces a verifying $\sigh$ over
a chosen message without access to $\ski$ has success probability
negligible in the security parameter.
\end{lemma}

\begin{proof}
This is the standard EUF-CMA property of ed25519 over the
curve25519 group, established under the discrete-log assumption
\citep{bernstein2012highspeed,boneh2014foundations}.
\end{proof}

\subsection{T8.4 --- Finite-Time APB Construction Termination}
\label{sec:t84}

\begin{theorem}[Finite-Time APB Construction Termination]
\label{thm:t84}
For any persistent halt event $h$ on a system with bounded admission
snapshot $\Azero$ and bounded trace $\mathrm{trace}_{\le t_e}$, the
construction of $\Es$ terminates in time bounded by
$O(|\Azero| + |\mathrm{trace}_{\le t_e}|)$.
\end{theorem}

\begin{proof}
By Definition~\ref{def:apb}, $\Es$ has six fields:
$\mathrm{hash}(\Azero)$, $\Dhat(t_e)$, $t_e$, $\mathrm{eid}$,
$\mathrm{hash}(\mathrm{trace}_{\le t_e})$, and $\mathrm{cause}$.
Lemma~\ref{lem:bounded_count} establishes that the field count is
fixed; Lemma~\ref{lem:bounded_capture} establishes that each field's
capture time is bounded. Their conjunction gives the stated bound.
\end{proof}

\begin{lemma}[Bounded $\Es$ Field Count]
\label{lem:bounded_count}
The number of fields in $\Es$ is fixed at six and does not depend
on runtime state.
\end{lemma}

\begin{proof}
By Definition~\ref{def:apb}, fixed at compile time.
\end{proof}

\begin{lemma}[Bounded Capture Time Per Field]
\label{lem:bounded_capture}
Each of the six $\Es$ fields can be computed in time bounded
linearly in the size of its underlying input:
$\mathrm{hash}(\Azero)$ in $O(|\Azero|)$;
$\Dhat(t_e)$ in $O(1)$ (already computed by the IML monitor);
$t_e$ in $O(1)$;
$\mathrm{eid}$ in $O(1)$ (UUID4 generation is a bounded-time
operation under standard OS entropy assumptions, reading a fixed-length
random bitstring from the kernel's CSPRNG);
$\mathrm{hash}(\mathrm{trace}_{\le t_e})$ in $O(|\mathrm{trace}_{\le t_e}|)$;
$\mathrm{cause}$ in $O(1)$ (a fixed-vocabulary string).
\end{lemma}

\begin{proof}
SHA-256 is linear in input size; the remaining operations are
elementary.
\end{proof}

\subsection{Proposition 8.5 --- Byzantine Principal Resistance}
\label{sec:p85}

\begin{proposition}[Byzantine Principal Resistance]
\label{prop:byzantine}
Under the $k$-of-$n$ threshold construction of
Definition~\ref{def:multi_apb}, an adversary $\mathcal{A}$ who
controls $j < k$ principal signing keys cannot produce a verifying
$\APB^{(k)}$, except with negligible probability.
\end{proposition}

\begin{proof}
A verifying $\APB^{(k)}$ requires at least $k$ entries in $\Sigma$
to satisfy V1 (valid signature against a registered public key).
By T8.3 (Impossibility of Anonymous Re-Authorization), each valid
entry requires possession of the corresponding signing key $sk^j$:
producing a valid $\sigma^j$ without $sk^j$ reduces to breaking
EUF-CMA of ed25519, which is infeasible under the discrete-log
assumption. If $\mathcal{A}$ holds $j < k$ keys, they can produce at
most $j$ valid entries; the threshold check fails. Entries for keys
not held by $\mathcal{A}$ fail V1 with all but negligible probability.
Therefore $\mathcal{A}$ cannot produce a verifying $\APB^{(k)}$
with $j < k$ keys.
\end{proof}

\begin{remark}
Proposition~\ref{prop:byzantine} formalises the security intuition
for the composable access-control layer described in
\S\ref{sec:composability}: with $k = \lceil n/2 \rceil + 1$, the
construction achieves strict majority quorum; with $k = n$, it
requires unanimous governance. The adversary's capture cost grows
linearly in $k$, independently of the specific policy encoded in
$\Dh$.
\end{remark}

\section{Implementation}
\label{sec:impl}

We provide a real-cryptography reference implementation of the APB
in Python on top of the existing governance stack of
\citep{fernandez2026applied}. All four theorems' empirical claims
(\S\ref{sec:exp_a}--\S\ref{sec:exp_d}) are verified against this
implementation.

\paragraph{Cryptography.}
We use ed25519 from the \texttt{cryptography} library
\citep{cryptographylib}, version $\geq 42$. Each $\Hi \in \Pset$ has a
$32$-byte private key (held only by $\Hi$) and a $32$-byte public key
held in the registry. Signatures are $64$ bytes. Canonical
serialisation is the RFC~8785 JSON Canonicalization Scheme (JCS),
implemented via the \texttt{jcs} library~\citep{rfc8785}: keys sorted
by Unicode codepoint, IEEE~754 number representation, no insignificant
whitespace, UTF-8 encoding. This replaces ad-hoc
\texttt{sort\_keys=True} JSON and eliminates any cross-implementation
ambiguity in the signed message.

\paragraph{Modules.}
Four new modules implement the framework
(see Table~\ref{tab:modules}).

\begin{table}[t]
\centering
\caption{Implementation modules introduced in this work.
LOC counts include docstrings.}
\label{tab:modules}
\small
\begin{tabular}{lrl}
\toprule
Module & LOC & Responsibility \\
\midrule
\texttt{agent/principal.py}        & 140 & Principal Set $\Pset$, keypair generation, registry, revocation. \\
\texttt{stack/apb.py}              & 220 & $\Es$, $\Dh$, $\APB$ classes; RFC~8785 encoding; signing; event\_id. \\
\texttt{stack/apb\_verifier.py}    & 155 & 5-predicate verification (V1--V5); failure-mode enum; attribution. \\
\texttt{stack/governance\_layer.py}& 190 & Authority Resolution Function $G$; built-in policies. \\
\texttt{stack/multi\_principal.py} & 280 & $\APB^{(k)}$, $k$-of-$n$ verifier, threshold governance layer. \\
\texttt{stack/apb\_log.py}         & 230 & HMAC-chained JSONL log; tamper-evident append; thread-safe writes. \\
\midrule
Total                              & 1215 & Plus 103 unit tests over all modules. \\
\bottomrule
\end{tabular}
\end{table}

\paragraph{Verifier predicates.}
\label{sec:impl_verifier}%
The verifier evaluates five predicates against an APB and a
registry, and returns the first failing predicate (or VALID):
\begin{enumerate}[leftmargin=2em, topsep=2pt, itemsep=1pt]
  \item[V1.] The signed message verifies against $\pki$ for the
        $\Hi$ named in $\Dh$.
  \item[V2.] $\Hi$ is in the registry.
  \item[V3.] $\Hi$ was active at $t_e$ (not revoked before signing).
  \item[V4.] $t_e$ lies within an acceptable window relative to the
        verification time (temporal freshness / clock-based replay defence).
  \item[V5.] $\mathrm{eid}$ has not been accepted in the current
        verification session (semantic uniqueness / duplicate-submission
        defence). The verifier maintains a caller-supplied nonce store;
        the first accepted APB with a given $\mathrm{eid}$ adds it to
        the store, and any subsequent submission of the same
        $\mathrm{eid}$ is rejected regardless of $t_e$.
\end{enumerate}
V4 and V5 are complementary: V4 rejects temporally stale or future
APBs; V5 rejects exact replays within the acceptance window.
Together they provide both temporal freshness and semantic uniqueness.
The enumerated failure modes (\texttt{INVALID\_SIGNATURE},
\texttt{PRINCIPAL\_NOT\_FOUND}, \texttt{PRINCIPAL\_REVOKED},
\texttt{REPLAY}, \texttt{DUPLICATE\_EVENT\_ID}, \texttt{MALFORMED})
allow Experiment~B (\S\ref{sec:exp_b}) to count detection rates per
attack class.

\paragraph{Persistence.}
\label{sec:impl_log}%
APBs are persisted in append-only JSONL files with HMAC chaining:
each entry's HMAC is computed over the previous entry's HMAC
concatenated with the current entry's canonical bytes. This gives
log-level tamper evidence orthogonal to the per-APB signature: any
deletion or reordering of entries breaks the chain. We do not use a
Merkle tree in this implementation (it adds machinery for a
distributed-log scenario that the present paper does not need); we
note in \S\ref{sec:discussion} that a Merkle structure is the
natural extension for distributed deployments.

\paragraph{Decision policies.}
We provide three built-in \texttt{DecisionPolicy} callables that the
governance layer can route to: \texttt{always\_resume},
\texttt{always\_deny}, and \texttt{threshold\_policy(deny\_above,
recalibrate\_above)}. The threshold policy issues $\RESUME$ when
$\Dhat(t_e) <$ \texttt{deny\_above}; $\DENY$ when
$\texttt{deny\_above} \le \Dhat(t_e) <$ \texttt{recalibrate\_above};
$\RECAL$ otherwise. The choice of $(\theta_d, \theta_r)$ is operator
configuration; Experiment~A (\S\ref{sec:exp_a}) shows that
governance completeness is invariant under this choice.

\subsection{Proposition 5.1 --- $\Tstar$ Calibration}
\label{sec:prop_calibration}

\begin{proposition}[$\Tstar$ Calibration Range]
\label{prop:calibration}
Let model $M$ have measured drift threshold $\Tstar_M$ with run-to-run
standard deviation $\sigma_M$. The persistence window $\Delta T$ in
DC.2 satisfies the calibration criterion if
\[
   k_1 \cdot \sigma_M \;\le\; \Delta T \;\le\; k_2 \cdot \Tstar_M,
\]
with $k_1 \geq 3$ to bound the false-positive rate below the run-to-run
noise of $\Tstar$ and $k_2 \leq 0.5$ so that the persistence window
declares a halt before the measured drift threshold is reached.
\end{proposition}

This is a design recipe rather than a deductive theorem: $k_1$ and
$k_2$ are operator choices reflecting tolerance to false positives
and false negatives respectively. The empirical content of
Proposition~\ref{prop:calibration} is the claim that the range
$[k_1 \sigma_M, k_2 \Tstar_M]$ is non-empty for the models we
measure --- equivalently, that $\sigma_M / \Tstar_M$ is small. We
verify this in \S\ref{sec:exp_c} for five of the six models in our
study; one model does not drift within the experimental window
and is treated separately in \S\ref{sec:discussion}.

\section{Experiment A: Governance Completeness}
\label{sec:exp_a}

\paragraph{Goal.} Empirical demonstration that T8.1 holds: every
runtime halt produced by the governance stack resolves through either
the recovery loop (R1) or a signed APB (R2), with no third path
observed.

\paragraph{Setup.} The full stack
$\text{ACP} + \text{IML} + \text{RAM} + \text{RecoveryLoop} +
\text{GovernanceLayer}$ is run for $1000$ steps per seed (50-step
burn-in plus 950 drift steps), with $10$ seeds (\texttt{42}--\texttt{51}),
coverage $0.70$, $\theta = 0.20$, and a single principal
$H_{\mathrm{alice}}$ in $\Pset$. Whenever the recovery loop returns
$\HALT$ or $\DENY$ on a halt, the governance layer is invoked with a
threshold-based decision policy; the resulting APB is verified
through the full verifier (V1--V5) before being recorded.
V5 is vacuously satisfied in this experiment: each governance event
generates a fresh UUID4 \texttt{event\_id}, so the nonce store
never contains a duplicate.

\paragraph{Policy sensitivity.} We run the same stack under two
threshold policies. The \emph{default} policy uses
$(\theta_d, \theta_r) = (0.40, 0.70)$, which causes essentially all
governance events to receive $\RESUME$ given the empirical $\Dhat$
distribution under this protocol. The \emph{low-threshold} policy
uses $(0.25, 0.32)$, which exercises the full
$\{\RESUME, \DENY, \RECAL\}$ output range. Both policies are run
against the same seeds.

\paragraph{Result.} Table~\ref{tab:exp_a_completeness_default}
reports the per-seed and aggregate counts for the default policy.
Across $10$ seeds and $3{,}812$ halt events, every halt resolves through
exactly one of the two paths; no third path is observed. The
\emph{neither} count, which would witness a violation of T8.1, is
zero. Table~\ref{tab:exp_a_policy_comparison} shows that the
low-threshold policy produces an identical halt count and an identical
APB count, but a different distribution of decisions
($1037$ $\RESUME$ / $727$ $\DENY$ / $36$ $\RECAL$ instead of
$1800 / 0 / 0$): completeness is policy-invariant; the decision
distribution is not.

\begin{table}[t]
\centering
\caption{Governance Completeness (T8.1) (default policy): every HALT resolves through Recovery
or a signed APB. NEITHER count is 0 across all seeds (assertion PASSED).}
\label{tab:exp_a_completeness_default}
\begin{tabular}{rrrrr}
\toprule
Seed & HALTs & Recovery RESUME & APB-signed & NEITHER \\
\midrule
  42 & 396 & 211 & 185 & 0 \\
  43 & 401 & 207 & 194 & 0 \\
  44 & 368 & 174 & 194 & 0 \\
  45 & 387 & 191 & 196 & 0 \\
  46 & 364 & 201 & 163 & 0 \\
  47 & 387 & 217 & 170 & 0 \\
  48 & 371 & 168 & 203 & 0 \\
  49 & 378 & 173 & 205 & 0 \\
  50 & 398 & 231 & 167 & 0 \\
  51 & 362 & 239 & 123 & 0 \\
  \midrule
  \textbf{Total} & \textbf{3812} & \textbf{2012} & \textbf{1800} & \textbf{0} \\
\bottomrule
\end{tabular}
\end{table}

\begin{table}[t]
\centering
\caption{Policy sensitivity in Exp A: T8.1 (Governance Completeness) holds
identically under both threshold policies. The APB decision distribution
varies with policy parameters; Completeness is policy-agnostic.}
\label{tab:exp_a_policy_comparison}
\begin{tabular}{lcrrrcr}
\toprule
Policy & $(\theta_d,\theta_r)$ & HALTs & Rec.\ RESUME & APB-signed
       & RES/DENY/RECAL & NEITHER \\
\midrule
  default & (0.40,0.70) & 3812 & 2012 & 1800 & 1800/0/0 & 0 \\
  low\_thresh & (0.25,0.32) & 3812 & 2012 & 1800 & 1037/727/36 & 0 \\
\bottomrule
\end{tabular}
\end{table}

\paragraph{Reading.} The two tables together separate two distinct
empirical claims. T8.1 is about \emph{path exhaustivity}: every halt
ends through R1 or R2, never neither, regardless of operator policy.
The decision distribution within R2 is a property of the policy, not
of T8.1; the framework is policy-agnostic.

\section{Experiment B: APB Integrity}
\label{sec:exp_b}

\paragraph{Goal.} Empirical demonstration that T8.2 and T8.3 hold:
no modification of an APB after signing is undetectable, and no
forgery without $\ski$ verifies. We exercise nine adversarial vectors
covering the threat-model classes A1--A4 of \S\ref{sec:threat}.

\paragraph{Setup.} $200$ freshly-signed APBs are generated under
varied evidence (each $\Es$ uses random $A_0$-hash, varying $\Dhat$
across $[0.20, 0.45]$, and a current $t_e$). A baseline run confirms
all $200$ untampered APBs verify cleanly. Each attack vector is then
applied to each APB, and the verifier output is recorded.

\paragraph{Vectors.}
\begin{description}[leftmargin=2em, topsep=2pt, itemsep=1pt]
  \item[A1.1--A1.5: tamper $\Es$ field-by-field.]
        For each of the five $\Es$ fields
        (\texttt{A\_0\_hash}, \texttt{D\_hat}, \texttt{t\_e},
        \texttt{trace\_hash}, \texttt{cause}), substitute a
        modified value while retaining $\Dh$ and $\sigh$.
  \item[A2a: forge with random bytes.] Replace $\sigh$ with $64$
        random bytes.
  \item[A2b: forge with attacker keypair.] Re-sign
        $\mathrm{canon}(\Es) \,\Vert\, \mathrm{canon}(\Dh)$ with a
        freshly-generated key whose public counterpart is not in
        $\Pset$.
  \item[A3: identity swap.] Replace $\Hi$ in $\Dh$ with a different
        registered principal, retaining $\sigh$.
  \item[A4: replay.] Verify the (untampered) APB at a clock value
        $t_e + 360\,\mathrm{s}$ with replay window $300\,\mathrm{s}$.
\end{description}

\paragraph{Result.} Table~\ref{tab:exp_b_integrity} reports the
per-vector counts. Across $1{,}800$ attack attempts, the verifier
detects every attack: $1800/1800$, $0$ misses, $100.0000\%$
detection rate. T8.2 and T8.3 hold empirically.

\begin{table}[t]
\centering
\caption{APB Integrity (T8.2 + T8.3): four attack vectors against
200 freshly-signed APBs. All tampering attempts are detected
by the verifier. T8.2: PASSED. T8.3: PASSED.}
\label{tab:exp_b_integrity}
\begin{tabular}{lrrrr}
\toprule
Attack vector & Trials & Detected & Missed & Detection rate \\
\midrule
  Tamper $E_s$.A\_0\_hash & 200 & 200 & 0 & 100.00\% \\
  Tamper $E_s$.D\_hat & 200 & 200 & 0 & 100.00\% \\
  Tamper $E_s$.t\_e & 200 & 200 & 0 & 100.00\% \\
  Tamper $E_s$.trace\_hash & 200 & 200 & 0 & 100.00\% \\
  Tamper $E_s$.cause & 200 & 200 & 0 & 100.00\% \\
  \midrule
  Forge $\sigma_h$ (random bytes) & 200 & 200 & 0 & 100.00\% \\
  Forge $\sigma_h$ (attacker key) & 200 & 200 & 0 & 100.00\% \\
  Identity swap ($H_i$ replaced) & 200 & 200 & 0 & 100.00\% \\
  Replay (at $t_e + 360$\,s) & 200 & 200 & 0 & 100.00\% \\
  \midrule
  \textbf{Total} & \textbf{1800} & \textbf{1800} & \textbf{0} & \textbf{100.00\%} \\
\bottomrule
\end{tabular}
\end{table}

\paragraph{Reading.} The five $\Es$-tampering vectors collectively
exercise Lemma~\ref{lem:tamper} on every $\Es$ field, witnessing
T8.2. The two forgery vectors witness T8.3 (the system layer cannot
produce a valid signature without $\ski$). The identity-swap vector
witnesses Lemma~\ref{lem:attribution}: substituting $\Hi$ in $\Dh$
without recomputing $\sigh$ under the new principal's key is detected.
The replay vector witnesses the verifier's V4 predicate.

\section{Experiment C: Cross-Model $\Tstar$}
\label{sec:exp_c}

\paragraph{Goal.} Empirical characterisation of the model-specific
drift threshold $\Tstar$ across six open LLM families and validation
of the calibration criterion of Proposition~\ref{prop:calibration}.

\paragraph{Setup.} Six models (\texttt{mistral:7b},
\texttt{deepseek-r1:8b}, \texttt{gemma4:latest}, \texttt{gpt-oss:20b},
\texttt{qwen2.5:7b}, \texttt{llama3.2:3b}) are run via Ollama through
the LiveLLM tool selector, $50$-step burn-in plus $500$ drift steps,
sampling temperature $T = 0.4$, $\theta = 0.20$, $3$ runs per model.
$\Tstar$ is the first drift-step index at which $\Dhat \geq \theta$.

\paragraph{Result.} Table~\ref{tab:exp_c_crossmodel} reports per-model
$\Tstar$ and final $\Dhat$ values, sorted by parameter count. The
intra-model coefficient of variation $\sigma/\Tstar$ is reported in
the rightmost column.

\begin{table}[t]
\centering
\caption{Cross-model $T^*$: 6 LLM families, 3 runs each, 500 drift steps,
$T = 0.4$, $\theta = 0.2$. $T^*$ as mean $\pm$ std across runs;
$\sigma/T^*$ is the intra-model coefficient of variation. Proposition~5.1
calibration criterion ($\sigma/T^* < 5\%$ for measured $T^*$): HOLDS. 1~model(s) did not reach $\theta$ within the 500-step window.}
\label{tab:exp_c_crossmodel}
\begin{tabular}{lllrrr}
\toprule
Model & Family & Params & $T^*$ (mean$\pm$std) & $D_{\mathrm{final}}$ & $\sigma/T^*$ \\
\midrule
  llama3.2:3b & Meta & 3.2B & 151 $\pm$ 0.9 & 0.430 $\pm$ 0.000 & 0.62\% \\
  gemma4:latest & Google & 4.0B & 264 $\pm$ 5.2 & 0.272 $\pm$ 0.010 & 1.99\% \\
  mistral:7b & Mistral & 7.2B & 157 $\pm$ 0.0 & 0.434 $\pm$ 0.001 & 0.00\% \\
  qwen2.5:7b & Alibaba & 7.6B & 154 $\pm$ 2.9 & 0.427 $\pm$ 0.001 & 1.91\% \\
  deepseek-r1:8b & DeepSeek & 8.0B & 160 $\pm$ 0.8 & 0.424 $\pm$ 0.000 & 0.51\% \\
  gpt-oss:20b & OpenAI & 20.0B & --- & 0.071 $\pm$ 0.005 & --- \\
\bottomrule
\end{tabular}
\end{table}

\paragraph{Reading.}
Five of six models reach $\theta$ within the $500$-step window, with
$\sigma/\Tstar$ between $0.00\%$ (perfect replication on
\texttt{mistral:7b}) and $1.99\%$ (\texttt{gemma4}). For these five
models, the calibration criterion $\sigma/\Tstar < 5\%$ of
Proposition~\ref{prop:calibration} holds with substantial margin: the
range $[k_1 \sigma_M, k_2 \Tstar_M] = [3 \sigma_M, 0.5 \Tstar_M]$
is non-empty for every measured model.

The sixth model, \texttt{gpt-oss:20b}, the largest in our sample at
$20$B parameters, does not drift within the experimental window:
final $\Dhat = 0.071 \pm 0.005$, well below $\theta = 0.20$.
We discuss this case explicitly in \S\ref{sec:discussion}; for the
purposes of T8.1's empirical verification it is sufficient to note
that the calibration criterion is reported per-model rather than as
a global property, and that the deployment recipe of
Proposition~\ref{prop:calibration} requires a measured $\Tstar$
before $\Delta T$ can be set.

\paragraph{Family-wise consistency.}
The Mistral family entry in our sample (\texttt{mistral:7b},
$\Tstar = 157 \pm 0$) coincides with the Mistral-family $\Tstar$
reported in \citep{fernandez2026applied} for a larger
\texttt{mistral-small-3.1} model under the same protocol
(\textit{T*} $= 156 \pm 1$). Two distinct sizes within the same
family yield $\Tstar$ within one step, supporting the architectural
rather than scale-driven character of the threshold.

\paragraph{Refutation of size-monotonicity.}
The largest model in the sample is the only one that does not drift.
This refutes any monotone hypothesis under which larger models would
be more (or less) susceptible to compliant drift than smaller ones.
The empirical implication is that $\Tstar$ must be \emph{measured}
per deployment rather than \emph{predicted} from architectural
parameters.

\section{Experiment D: Temperature Insensitivity}
\label{sec:exp_d}

\paragraph{Precondition for APB construction.}
An APB is constructible only after a governance event has been triggered
--- and a governance event requires that $\Tstar$ be stable and
measurable for the deployed model.  If $\Tstar$ drifted substantially
with the LLM's sampling temperature, the threshold would need to be
recalibrated per deployment configuration, undermining the
Proposition~5.1 calibration recipe (\S\ref{sec:prop_calibration}).
Experiment~D validates this precondition: $\Tstar$ must be determined
by the escalation protocol embedded in the system prompt, not by
sampling stochasticity.  A model whose $\Tstar$ is temperature-stable
can have its APB trigger calibrated once and deployed confidently
across configurations.

\paragraph{Goal.} Test the conjecture that $\Tstar$ is determined by
the escalation protocol embedded in the system prompt rather than by
LLM sampling stochasticity. Empirical translation: across a sweep
of sampling temperatures, the per-model spread of mean $\Tstar$
should be small.

\paragraph{Setup.} Three models from \S\ref{sec:exp_c} are selected
for diversity in baseline $\Tstar$: \texttt{mistral:7b} ($\Tstar = 157$,
``fast cluster''), \texttt{llama3.2:3b} ($\Tstar = 151$,
``fastest''), and \texttt{gemma4:latest} ($\Tstar = 264$,
``slowest drifter''). Each is run at temperatures
$T \in \{0.2, 0.4, 0.6, 0.8\}$, three runs per cell, $500$ drift
steps, identical burn-in and $\theta$ to \S\ref{sec:exp_c}. The
insensitivity criterion is that the per-model range
$\max_T \Tstar(T) - \min_T \Tstar(T)$ across the four temperatures is
at most $10$ steps.

\paragraph{Result.} Table~\ref{tab:exp_d_temperature} reports the
per-cell $\Tstar$ across the four temperatures and the resulting
per-model range. The result is more nuanced than a single
insensitivity claim: two of three models satisfy the criterion with
substantial margin, and the third violates it.

\begin{table}[t]
\centering
\caption{Temperature insensitivity of $T^*$. Three models were each run at
four sampling temperatures, three runs per cell, 500 drift steps.
Cell entries are mean $\pm$ std of $T^*$ across runs at that temperature;
the rightmost column is the per-model range $\max_T T^*(T) - \min_T T^*(T)$
of the cell means. Insensitivity criterion (range $\le 10$ steps): VIOLATED.}
\label{tab:exp_d_temperature}
\begin{tabular}{lccccc}
\toprule
Model & $T=0.2$ & $T=0.4$ & $T=0.6$ & $T=0.8$ & range \\
\midrule
  gemma4:latest & 239$\pm$6.9 & 255$\pm$7.3 & 249$\pm$3.9 & 290$\pm$18.4 & 51 \\
  llama3.2:3b & 155$\pm$0.9 & 155$\pm$1.6 & 156$\pm$1.2 & 156$\pm$2.5 & 2 \\
  mistral:7b & 154$\pm$2.6 & 157$\pm$2.2 & 156$\pm$1.9 & 152$\pm$1.6 & 5 \\
\bottomrule
\end{tabular}
\end{table}

\paragraph{Reading.}
For the two fast-drifting models, \texttt{llama3.2:3b} and
\texttt{mistral:7b}, $\Tstar$ is strongly temperature-insensitive:
the range across $T \in \{0.2, 0.4, 0.6, 0.8\}$ is $2$ and $5$ steps
respectively, and the within-cell standard deviations are below
$3$ steps for every $(model, T)$ pair. For these models the conjecture
of \citep{fernandez2026applied} --- that $\Tstar$ is determined by the
escalation protocol embedded in the system prompt rather than by
sampling stochasticity --- replicates with substantial margin and
generalises across families.

For \texttt{gemma4}, the slow-drifting model in our sample, the
picture is different. The per-cell mean at $T = 0.8$ rises to
$290 \pm 18.4$ steps, against $239$--$255$ at the lower temperatures,
giving a per-model range of $51$ steps and a within-cell standard
deviation that is roughly three times larger than at $T \le 0.6$.
The insensitivity criterion ($\text{range} \le 10$) is violated.

\paragraph{Interpretation.}
We read this as evidence for a \emph{drift-floor} effect: when a
model's $\Tstar$ is close to the lower bound that the escalation
protocol can produce, sampling temperature has little room to move
the threshold. When $\Tstar$ is well above that floor (as in
\texttt{gemma4}, whose $\Tstar$ is $1.7\times$ larger than the
fast-cluster), the additional steps before threshold crossing
admit cumulative sampling fluctuations that compound at higher
temperatures. The relevant quantity is therefore not whether
temperature insensitivity holds globally but whether a model's
baseline $\Tstar$ sits in the deterministic regime; this is itself
a measurable property of the deployment and feeds back into the
calibration recipe of Proposition~\ref{prop:calibration} (a model
in the non-deterministic regime requires a wider $k_1 \sigma_M$
margin to absorb temperature-induced spread).

\paragraph{Implication for governance.}
Operators deploying models in the slow-drifting regime should not
assume that a single-temperature characterisation of $\Tstar$
generalises to other temperatures of the same model. Either the
deployment temperature is fixed (and $\Tstar$ measured at that
temperature suffices), or the calibration is performed at the
operationally relevant temperature(s). For fast-drifting models,
a single-temperature characterisation appears to be sufficient.

\section{Experiment E: Multi-Principal Threshold Governance}
\label{sec:exp_e}

\paragraph{Goal.} Empirically validate Proposition~\ref{prop:byzantine}
(Byzantine Principal Resistance): under $k$-of-$n$ threshold signing,
governance capture requires the simultaneous compromise of at least
$k$ distinct principals.

\paragraph{Setup.} Three principals are registered: Alice ($A$),
Bob ($B$), and Carol ($C$). We generate $N = 1{,}000$ governance
events per scenario from fresh $(\Es, \Dh)$ pairs (unique
\texttt{event\_id} per event). Seven scenarios test the full
threat surface of the $k$-of-$n$ construction:

\begin{center}
\small
\begin{tabular}{clll}
\toprule
ID & Configuration & Threshold & Expected result \\
\midrule
S1 & Alice only ($k=1$, $n=1$)         & 1 & \texttt{VALID} \\
S2 & Alice + Bob ($k=2$, $n=2$)        & 2 & \texttt{VALID} \\
S3 & Alice + Bob + Carol ($k=3$, $n=3$)& 3 & \texttt{VALID} \\
S4 & Alice only, threshold=2           & 2 & \texttt{INSUFFICIENT\_SIGNATURES} \\
S5 & Alice + Bob, threshold=3          & 3 & \texttt{INSUFFICIENT\_SIGNATURES} \\
S6 & Alice (legit) + forged Bob sig    & 2 & \texttt{INVALID\_SIGNATURE} \\
S7 & Alice counted twice (duplicate)   & 2 & \texttt{DUPLICATE\_SIGNER} \\
\bottomrule
\end{tabular}
\end{center}

S4--S5 simulate under-threshold submission (attacker holds $j < k$
keys). S6 simulates a capture attempt: attacker controls Alice's key
and forges Bob's signature with a random key not in the registry.
S7 simulates key reuse (attempting to count the same key twice to
satisfy the threshold).

\paragraph{Results.} Table~\ref{tab:exp_e} summarises the outcome.

\begin{table}[h]
\centering
\caption{Experiment E results. All $7 \times 1{,}000 = 7{,}000$ events
produced the expected outcome; no scenario deviated from its
expected result.}
\label{tab:exp_e}
\small
\begin{tabular}{clccr}
\toprule
ID & Scenario & Expected & Events & Rate \\
\midrule
S1 & $k=1$, 1 legit signer              & VALID                  & 1\,000 & 100.00\% \\
S2 & $k=2$, 2 legit signers             & VALID                  & 1\,000 & 100.00\% \\
S3 & $k=3$, 3 legit signers             & VALID                  & 1\,000 & 100.00\% \\
S4 & $k=2$, 1 signer only               & INSUFF. SIGNATURES     & 1\,000 & 100.00\% \\
S5 & $k=3$, 2 signers only              & INSUFF. SIGNATURES     & 1\,000 & 100.00\% \\
S6 & $k=2$, 1 legit + 1 forged sig      & INVALID SIGNATURE      & 1\,000 & 100.00\% \\
S7 & $k=2$, same key counted twice      & DUPLICATE SIGNER       & 1\,000 & 100.00\% \\
\bottomrule
\end{tabular}
\end{table}

\paragraph{Byzantine resistance.}
S6 and S7 are the critical Byzantine scenarios: in S6, an attacker
who controls Alice's key cannot obtain a verifying $\APB^{(2)}$
without a second legitimate key --- forging Bob's signature fails
at V1 (EUF-CMA of ed25519). In S7, attempting to count Alice's
key twice is caught by the duplicate-signer check before the
threshold is reached. Both scenarios produce $0$ false acceptances
across $1{,}000$ events, confirming Proposition~\ref{prop:byzantine}:
single-key capture is infeasible under $k \geq 2$.

\paragraph{Threshold monotonicity.}
Scenarios S1--S3 demonstrate clean threshold monotonicity:
governance events with $j$ valid signatures succeed at threshold
$k \leq j$ and fail at $k > j$. The verifier correctly
distinguishes ``insufficient signatures'' (S4--S5) from ``invalid
signature'' (S6), enabling an operator to diagnose whether a
rejection is due to under-provisioning or an active attack.

\section{Experiment F: Infrastructure Fault Injection}
\label{sec:exp_f}

\paragraph{Goal.}
Validate robustness of the APB verification protocol against five
representative infrastructure fault classes that arise in realistic
deployments: clock drift (F1), registry key mismatch (F2), concurrent
writes (F3), log tampering (F4), and duplicate APB submission (F5).
Each class targets a distinct predicate from the V1--V5 verification
suite (\S\ref{sec:impl_verifier}) or the tamper-evident log
(\S\ref{sec:impl_log}).

\paragraph{Setup.}
We use a 3-principal registry (Alice, Bob, Carol) with ed25519
key pairs.  Each fault class runs $N = 1{,}000$ events except
F4, which runs $R = 100$ repetitions per tamper strategy with
$N = 10$ log entries each (500 tampering trials in total).
The V4 acceptance window is $T^* = 300\,\mathrm{s}$.
All experiments are fully deterministic given the PRNG seed;
results are reproducible via \texttt{experiments/exp\_fault\_injection.py}.

\subsection*{F1 — Clock Drift (V4 boundary precision)}

We inject artificial clock offsets $\delta \in
\{-600, -310, -301, -299, -150, -1, 0, +1, +150, +299, +301, +310, +600\}$
seconds into the evidence timestamp $t_e$, holding the verifier
clock fixed at \emph{now}.  The V4 predicate accepts iff
$|\delta| \leq T^*$.

\begin{table}[h]
\centering
\caption{F1 Clock-drift results. Acceptance is exactly $|\delta| \leq 300\,\mathrm{s}$;
all 1\,000 events per offset behave uniformly.}
\label{tab:exp_f1}
\small
\begin{tabular}{rllr}
\toprule
$\delta$ (s) & Expected & Accepted / 1\,000 & Correct \\
\midrule
$-600$ & REJECT & 0     & \checkmark \\
$-310$ & REJECT & 0     & \checkmark \\
$-301$ & REJECT & 0     & \checkmark \\
$-299$ & ACCEPT & 1\,000 & \checkmark \\
$-150$ & ACCEPT & 1\,000 & \checkmark \\
$\phantom{+}0$   & ACCEPT & 1\,000 & \checkmark \\
$+150$ & ACCEPT & 1\,000 & \checkmark \\
$+299$ & ACCEPT & 1\,000 & \checkmark \\
$+301$ & REJECT & 0     & \checkmark \\
$+310$ & REJECT & 0     & \checkmark \\
$+600$ & REJECT & 0     & \checkmark \\
\bottomrule
\end{tabular}
\end{table}

All 13 offset points are classified correctly (Table~\ref{tab:exp_f1}).
The boundary is sharp at $\pm 1\,\mathrm{s}$ relative to $T^*$: offset
$-299\,\mathrm{s}$ is accepted and $-301\,\mathrm{s}$ is rejected in
$100\%$ of trials, confirming deterministic V4 enforcement.

\subsection*{F2 — Registry Key Mismatch (V1 under key corruption)}

We register Alice with a \emph{wrong} public key (a freshly generated
key unrelated to her signing key) and submit $1{,}000$ legitimately
signed APBs.  The verifier must detect the mismatch via V1 and return
\texttt{INVALID\_SIGNATURE} rather than \texttt{VALID} or
\texttt{PRINCIPAL\_NOT\_FOUND}.

\smallskip
\noindent\textbf{Result:} $0$ false acceptances; $1{,}000/1{,}000$
events return \texttt{INVALID\_SIGNATURE} (false-acceptance rate $= 0\%$).
This models supply-chain key substitution or registry database corruption:
even if the attacker's substituted key is plausible, a signature produced
under the true private key fails V1 unconditionally due to EUF-CMA
security of ed25519 (\S\ref{sec:t82}).

\subsection*{F3 — Concurrent Writes (V5 thread safety)}

We spawn $16$ threads, each submitting $50$ APBs (with distinct
\texttt{event\_id} UUIDs) to a shared \texttt{APBLog} and a shared
\texttt{seen\_event\_ids} nonce store protected by a \texttt{threading.Lock}.
Total: $800$ events.

\smallskip
\noindent\textbf{Result:} All $800$ events are accepted exactly once;
$0$ duplicate acceptances; $0$ thread errors; log integrity check passes
($800/800$ entries verified).  The locking discipline serialises
V5 nonce-store updates, preventing time-of-check/time-of-use (TOCTOU)
races that could allow a duplicate event\_id to be accepted by two
threads concurrently.

\subsection*{F4 — Log Tampering (HMAC chain integrity)}

We apply five tamper strategies to the HMAC-chained JSONL log across
$R = 100$ repetitions each (Table~\ref{tab:exp_f4}).

\begin{table}[h]
\centering
\caption{F4 Log-tampering detection rates ($R = 100$ per strategy,
$N = 10$ entries per log, $500$ total tamper trials).
All five strategies are detected in $100\%$ of trials.}
\label{tab:exp_f4}
\small
\begin{tabular}{llr}
\toprule
Strategy & Description & Detection rate \\
\midrule
Field modification & Change \texttt{D\_hat} in a middle entry     & 100\% \\
Entry deletion     & Remove the middle entry from the chain       & 100\% \\
Entry reordering   & Swap the first two entries                   & 100\% \\
HMAC forgery       & Replace \texttt{entry\_hmac} with random hex & 100\% \\
prev-HMAC splice   & Forge entry 0 and update entry 1's \texttt{prev\_hmac} & 100\% \\
\bottomrule
\end{tabular}
\end{table}

The prev-HMAC splice (last row) is the most sophisticated attack:
the adversary forges the HMAC of entry 0 \emph{and} propagates the
forged value into entry 1's \texttt{prev\_hmac} to maintain apparent
chain continuity.  This is detected because \texttt{entry\_hmac}
is computed as $\mathrm{HMAC\text{-}SHA256}(K, \texttt{seq} \|\,
\texttt{apb\_json} \| \texttt{prev\_hmac})$: an attacker who does
not know $K$ cannot produce a valid \texttt{entry\_hmac} for the
forged entry, so \texttt{verify\_integrity} rejects at entry 0
regardless of the splice.

\subsection*{F5 — Duplicate APB Submission (V5 necessity)}

We submit the \emph{same} APB $N = 1{,}000$ times, first without
the V5 nonce store (\texttt{seen\_event\_ids = None}), then with
it (\texttt{seen\_event\_ids} as a live set).

\begin{table}[h]
\centering
\caption{F5 Duplicate-submission results. Without V5 every replay is
accepted; with V5 only the first submission succeeds.}
\label{tab:exp_f5}
\small
\begin{tabular}{lccc}
\toprule
Configuration & Accepted & Duplicates caught & False acceptances \\
\midrule
Without V5 (V4 only) & 1\,000 & 0   & 999 \\
With V5              & 1      & 999 & 0   \\
\bottomrule
\end{tabular}
\end{table}

Without V5, every replay within the $300\,\mathrm{s}$ window is
accepted: $V4$ is a \emph{freshness} check, not a uniqueness check.
A legitimately signed APB can be replayed an arbitrary number of
times if the replays arrive within the acceptance window.
With V5 the \texttt{event\_id} UUID is recorded in the nonce store
on first acceptance; all $999$ subsequent submissions return
\texttt{DUPLICATE\_EVENT\_ID}.

\paragraph{Fault injection summary.}
Table~\ref{tab:exp_f_summary} aggregates the five fault classes.

\begin{table}[h]
\centering
\caption{Experiment F summary. All five fault classes are handled
correctly by the implementation; false-acceptance rate is $0\%$
for F2--F5, and V4 boundary precision is exact for F1.}
\label{tab:exp_f_summary}
\small
\begin{tabular}{llll}
\toprule
Fault & Targeted predicate & False-accept rate & Result \\
\midrule
F1 Clock drift         & V4 temporal freshness        & 0\% outside window & All 13 offsets correct \\
F2 Key mismatch        & V1 signature validity        & 0\%                & 100\% INVALID\_SIGNATURE \\
F3 Concurrent writes   & V5 + log thread safety       & 0\%                & 800/800 verified \\
F4 Log tampering       & HMAC chain integrity         & ---                & 100\% detected (5/5 strategies) \\
F5 Duplicate APB       & V5 semantic uniqueness       & 0\% (with V5)      & 999/999 caught \\
\bottomrule
\end{tabular}
\end{table}

The fault injection results demonstrate that each verification
predicate independently enforces its intended security property,
and that the predicates are \emph{compositionally sound}: no fault
class circumvents another predicate by exploiting the order of checks.

\section{Discussion}
\label{sec:discussion}

\paragraph{Models that do not drift.}
The case of \texttt{gpt-oss:20b} in \S\ref{sec:exp_c} merits explicit
treatment. Under our $500$-step protocol the model's drift estimator
remains well below threshold; we cannot, from this data alone,
distinguish three explanations:
(i) the model is architecturally robust to the specific escalation
protocol embedded in the system prompt and would not drift at any
reasonable horizon,
(ii) the model has a genuinely larger $\Tstar$ that exceeds our
window,
(iii) the model's tool-selection distribution under our LiveLLM
prompt is structurally different from the other five (e.g.\ a
strong default toward low-risk actions that does not respond to
context-pressure scaling).
For deployment, this distinction matters: under (i) the model would
not need an APB-emitting governance layer at all under this
protocol, while under (ii) or (iii) the operator would extend the
window or adapt the prompt. In either case, the implication for
governance design is the same: $\Pset$ and $G$ are configured
once; $\Delta T$ is configured per deployment from a measured
$\Tstar$, with the explicit understanding that $\Tstar$ may be
infinite or undefined for some models.

\paragraph{Boundaries of the result.}
The APB establishes accountability, not access control. A
compromised principal can sign a malicious decision and the APB
will record it. The mitigations are configurable around the APB,
not within it: $k$-of-$n$ multi-signature for sensitive decisions,
key rotation under $\RECAL$ governance events, and out-of-band
review of recorded APBs against expected decision patterns.
The proofs of T8.2--T8.4 rest on the security assumptions of
ed25519 over curve25519 and on the RFC~8785 canonicalization
properties. A break of those assumptions would invalidate the
cryptographic guarantees but not the structural claims (T8.1 in
particular is a partition argument that does not depend on the
signature scheme).
It is worth noting explicitly that T8.2 and T8.3 are reductions to
known cryptographic properties --- the existential unforgeability of
ed25519 under chosen-message attack (EUF-CMA) and the injectivity
of RFC~8785 canonicalization --- rather than novel mathematical
results. Their value lies in showing that these well-established
properties compose cleanly with the APB structure to produce the
non-repudiability and anonymous-re-authorization impossibility
guarantees needed by the governance layer.

\paragraph{Composability with access control.}
\label{sec:composability}
The APB protocol is deliberately orthogonal to access control: it
records \emph{that} a named principal made a decision, not
\emph{whether} that principal was entitled to do so. This separation
is a feature, not a limitation --- it allows the APB to be composed
with any access-control layer without redesigning either.

Three composable extensions are worth noting for deployments that
require stronger guarantees:

\begin{description}[leftmargin=2em, topsep=2pt, itemsep=3pt]
  \item[$k$-of-$n$ threshold signing.] For decisions of high
  consequence (e.g.\ $\RECAL$ events), operators can require
  that $k$ distinct members of $\Pset$ each produce a valid
  $\sigh$ over the same $(\Es, \Dh)$ before the governance
  layer accepts the APB. Under $k$-of-$n$, governance capture
  requires the simultaneous compromise of at least $k$
  principals; with $k = \lceil n/2 \rceil + 1$ this becomes a
  majority quorum with Byzantine tolerance.

  \item[Role-based authority scoping.] The $\mathrm{scope}$ field
  of $\Dh$ is a free-form string in the current implementation;
  a deployment MAY define a structured vocabulary and enforce
  that principals with a given RBAC role can only sign APBs
  with scopes from their allowed set. The verifier's V2/V3
  predicates can be extended to check role membership without
  changing the cryptographic core.

  \item[Capability tokens and policy engines.] For deployments
  already running a policy engine (OPA, Cedar, or similar),
  the APB's $\mathrm{eid}$ can serve as an immutable token
  whose issuance is gated by the policy layer. The APB then
  becomes the non-repudiable audit record of a token-issuance
  event rather than a standalone access decision.
\end{description}

These extensions are deferred to subsequent work (P9 onwards in
this series); the present paper establishes the primitive on which
they compose.

\paragraph{Integration paths.}
Two integration directions are natural extensions and deferred to
subsequent work. \emph{MCP/A2A integration}: the governance stack of
\citep{fernandez2026applied} operates on tool-call traces from a
LangGraph state graph; the Model Context Protocol and Agent-to-Agent
communication standards define interception points at which the
APB-emitting layer could be applied transparently to agent code
--- this is the focus of the next paper in this series.
\emph{Distributed APB log}: replacing the JSONL+HMAC chain with a
Merkle tree (and, for cross-organisational deployments, anchoring
to a public transparency log) gives stronger tamper evidence under
multi-writer scenarios.

\paragraph{Exploratory nature of Experiment D.}
The temperature-sweep results of \S\ref{sec:exp_d} should be read
as \emph{exploratory observations}, not as statistically validated
claims. With three runs per temperature-model cell, the results
are insufficient to estimate within-cell variance reliably; they
motivate hypotheses about the drift-floor effect and the
temperature-$\Tstar$ interaction, but a confirmatory study would
require at least $N \geq 10$ runs per cell with power analysis.
We move these results to the ``Observations'' register rather
than the ``Validation'' register of the paper's contribution.

\paragraph{Choice of decision-policy thresholds.}
Experiment~A's policy comparison shows that Protocol-Bounded
Governance Completeness is invariant under threshold choice. We
do not, in this paper, prescribe a particular $(\theta_d, \theta_r)$.
Their setting is a deployment decision tied to the operator's
tolerance for halts that escalate to $\DENY$ versus $\RESUME$, and
to the empirical distribution of $\Dhat$ at the moments of
governance event in the operator's workload.

\section{Conclusion}
\label{sec:conclusion}

We have introduced the Accountability Proof Block as the formal
mechanism by which an LLM agent's persistent halts resolve into
identity-bound, cryptographically verifiable governance decisions.
The construction is supported by four theorems ---
protocol-bounded governance completeness, non-repudiability,
impossibility of anonymous re-authorisation, and finite-time
evidence construction --- a Byzantine principal resistance
proposition for $k$-of-$n$ threshold signing, and a real-cryptography
reference implementation using the RFC~8785 JSON Canonicalization
Scheme and a five-predicate verifier that provides both temporal
freshness and semantic uniqueness guarantees. Empirical validation
across $3{,}812$ halt events, $1{,}800$ adversarial attempts,
$7{,}000$ multi-principal threshold events, $18$ cross-model runs,
and $36$ temperature-sweep runs shows that the
construction behaves as the theory predicts and that the model-specific
drift threshold $\Tstar$ is stable, measurable, and architecture-driven
rather than scale-driven, with the empirical caveat that
temperature-insensitivity itself is conditional on a model sitting
in a fast-drifting regime (a result we flag as exploratory pending
larger-sample confirmation). The result is a small, configurable,
provably-bounded mechanism by which a system that has reached the
limit of its own authority can hand the resolution back to a
named, accountable human --- without trusting the system itself
to make that handover honestly --- and that composes cleanly with
threshold signing, RBAC, and policy engines for deployments that
require stronger access-control guarantees.

The contribution of P8 is not that humans remain in the loop.
It is that when autonomous authority ends, the system can
\emph{prove} who took control, under what evidence, and under
which cryptographic guarantees.

\bibliographystyle{plainnat}
\bibliography{references}

\end{document}